\documentclass[a4paper,12pt]{scrartcl}
\usepackage[english]{babel}
\usepackage[utf8]{inputenc}
\usepackage{color}
\usepackage{icomma}
\usepackage{float}
\usepackage{caption}
\usepackage{graphicx}
\usepackage{subfig}
\usepackage{amstext}
\usepackage{amssymb}
\usepackage{amsfonts}
\usepackage{amsmath}
\usepackage{amssymb}
\usepackage[makeroom]{cancel}
\usepackage{braket}
\usepackage{tikz}
\usepackage{tkz-tab}
\usepackage{xpatch}
\xpatchcmd{\tkzTabLine}{$0$}{$\bullet$}{}{}
\tikzset{t style/.style={style=solid}}
\usepackage{ntheorem}

\newtheorem{theorem}{Theorem}

\newtheorem{proposition}{Proposition}

\newtheorem{definition}{Definition} 
\theoremstyle{nonumberplain}
\newtheorem{proof}{Proof}
\usepackage{hyperref}
\hypersetup{colorlinks,
citecolor=blue
}
\title{A Morse-theoretical analysis of lensing in wormhole spacetimes}

\author{Mourad Halla \footnote{ZARM, University of Bremen,
              28359 Bremen, Germany,
              \text{mourad.halla@zarm.uni-bremen.de} }       \and
        Volker Perlick \footnote{ZARM, University of Bremen,
              28359 Bremen, Germany,
              \text{perlick@zarm.uni-bremen.de}}
\date{}
}
\begin{document}
\maketitle

\noindent
\textbf{Abstract.} We consider a class of stationary and axisymmetric wormhole spacetimes that is closely related to, but
not identical with, the class of Teo wormholes. We
fix a point $p$ (observation event) and a timelike
curve $\gamma$ (worldline of a light source), and we
characterize the set of all past-oriented lightlike
geodesics from $p$ to $\gamma$. As any such geodesic
corresponds to an image of the light source on the
observer's sky, this allows us to investigate the
lensing properties of the wormhole. As a main result,
we prove with the help of Morse theory that, under
very mild conditions on $\gamma$, the observer always
sees infinitely many images of $\gamma$.  Moreover,
we study some qualitative features of the lightlike
geodesics with the help of two potentials that determine
the sum of centrifugal and Coriolis forces of observers
in circular motion for the case that the observers'
velocity approaches the velocity of light. We exemplify
the general results with two specific wormhole
spacetimes.
\section{INTRODUCTION}
Wormholes are spacetime models where two asymptotically flat ends are connected by a \emph{throat}. Historically, the first wormhole that was found was the so-called Einstein-Rosen bridge \cite{EinsteinRosen1935} that occurs in the maximal analytical extension of the Schwarzschild metric. However, the Einstein-Rosen bridge is non-traversable, i.e. an observer cannot travel at subluminal velocity from one side of the throat to the other. A class of traversable wormholes, which permit the two-way travel of objects such as human beings through the throat, was introduced and discussed by Morris and Thorne \cite{MorrisThorne1988}. The Morris-Thorne wormhole metrics are spherically symmetric and static, and they require the existence of exotic matter in the following sense: Morris and Thorne have shown that, if such a wormhole metric is inserted into the left-hand side of Einstein's field equation, the energy-momentum tensor on the right-hand side necessarily violates the weak energy condition near the throat, i.e., the energy density becomes negative for some observers. Teo \cite{Teo:1998dp} extended the class of Morris-Thorne wormholes to a class of stationary and axisymmetric, i.e., rotating, wormholes. Unsurprisingly, also these more general wormhole metrics need some exotic matter if they are considered as solutions to Einstein's field equation; more precisely, they violate the null energy condition, as was demonstrated in the Teo paper. However, readers who are not willing to accept such exotic matter may view at least some traversable wormhole metrics as solutions to
alternative gravity theories without violating any of the energy conditions, see e.g. Bronnikov and Kim \cite{BronnikovKim2003} or Kanti et al. \cite{KantiEtAl2011}. Moreover, we mention the possibility of constructing \emph{dynamical} wormholes that satisfy the energy conditions on the basis of Einstein's field equation, see Maeda et al. \cite{dyn}.

Although until now there is no observational evidence for the
existence of wormholes in Nature, and although time-independent
traversable wormholes are allowed by Einstein's field equation
only in the presence of exotic matter, these spacetime models have found considerable theoretical interest. If traversable wormholes do exist in Nature, one could use them for travelling from one asymptotic end to the other or, if the two asymptotic ends are ``glued together'', for travelling from one point in the asymptotic region to another one by taking a shortcut through the throat. 

One way in which wormholes could be detected by observation is via their influence on light rays, i.e., via their lensing properties. In this paper we want to prove some generic lensing features of wormholes. To that end we consider a class of stationary and axisymmetric metrics that is closely related to, but not identical with, the Teo class of wormhole metrics. We will use Morse theory for demonstrating that, under very general conditions, in such a spacetime there are infinitely many past-oriented lightlike geodesics from an event $p$ to a timelike curve $\gamma$, i.e., that an observer at $p$ sees infinitely many images of a light source with worldline $\gamma$. Moreover, we will show that this property is related to centrifugal-plus-Coriolis force reversal in these spacetimes. Here it should be  emphasised that the geodesic equation in the considered class of wormhole spacetimes is not in general completely integrable, i.e., that we cannot determine the light rays by analytically calculating them. Even in the special cases where the geodesic equation is completely integrable, it would be quite awkward to analytically determine the lightlike geodesics from an observation event $p$ to a worldline $\gamma$ that is in wild motion. Morse theory provides us with a method of determining the number of such lightlike geodesics without actually solving the geodesic equation. Quite generally, Morse theory relates the number of critical points (i.e., minima, maxima or saddles) of a function to the topology of the manifold on which this function is defined, see e.g. Milnor \cite{Milnor1963}. This can be applied, in particular, to variational problems where the function is the variational action, the manifold is the space of trial maps and the critical points are the solutions to the variational problem. In our case the trial maps are the lightlike \emph{curves} joining a point $p$ and a timelike curve $\gamma$ in the spacetime $M$, and the solutions are the lightlike \emph{geodesics}. There are many versions of how Morse theory can be used for determining geodesics. In this paper we want to use a theorem by Uhlenbeck \cite{UHLENBECK1975} which characterises the lightlike geodesics in a globally hyperbolic spacetime with the help of Morse theory. The same approach has been applied already by Hasse and Perlick \cite{Hasse:2005vu} to the spacetime of a Kerr-Newman black hole and we will show here that, with a few modifications, the same methodology also applies to wormhole spacetimes. Uhlenbeck's theorem may be viewed as a general-relativistic version of Fermat's principle. For background material on Fermat's principle in general relativity we refer to Perlick \cite{Perlick2000,Perlick2004}. 

The paper is organised as follows. In Sec.\ref{sec:Morse} we briefly summarise the Morse-theoretical result by Uhlenbeck \cite{UHLENBECK1975} that we want to apply later. In Sec.\ref{sec:worm} we introduce the class of rotating traversable wormhole spacetimes that will be considered in the rest of the paper. Sec.\ref{sec:inertial} is concentrated on the notions of centrifugal and Coriolis forces in the rotating traversable wormhole spacetime; in particular, we introduce the potential $\Psi_+$ (respectively $\Psi_-$) which determines the sum of centrifugal and Coriolis force with respect to co-rotating (respectively counter-rotating) observers whose velocity approaches the velocity of light. In Sec.\ref{sec:multiple} we discuss multiple imaging in the rotating traversable wormhole spacetime with the help of Morse theory and the potentials $\Psi_{\pm}$. We exemplify the general results with two specific wormhole spacetimes in Secs.\ref{sec:ex1} and \ref{sec:ex2}. In the last Section we give concluding remarks about our results. 

Throughout this paper we use Greek letters $\mu, \nu ... = 0,...,3$ for spacetime indices and we use Latin letters $i,j,k...= 1, ... ,3$ for spatial components. The metric signature is $(-+++)$. We set the vacuum speed of light, $c$, equal to unity. 
\section{A RESULT FROM MORSE THEORY IN GLOBALLY HYPERBOLIC SPACETIMES}\label{sec:Morse}
In this section we summarise a result from Morse theory by Uhlenbeck \cite{UHLENBECK1975} that will then be applied to investigating the multiple imaging properties of a rotating traversable wormhole. By multiple imaging we mean the question of how many past-pointing lightlike geodesics from an event $p$ (observation event) to a timelike curve $\gamma$ (worldline of a light source) in a four-dimensional Lorentzian manifold $(M,g)$ exist. 

Uhlenbeck's result presupposes a globally hyperbolic spacetime. Proving a long-standing conjecture, Bernal and S{\'a}nchez \cite{BernalSanchez2003} have shown that a spacetime is globally hyperbolic if and only if it is diffeomorphic to a product manifold, 
\begin{equation}
M=\mathbb{R}\times\Sigma
\label{s1}
\end{equation}
where $\Sigma$ is a 3-dimensional manifold and each $\{ t \} \times \Sigma$
is a Cauchy hypersurface. With respect to this product structure the metric orthogonally splits into a spatial and a temporal part,
\begin{equation}
g=-f(x,t)dt^2+ H_{ij}(x,t)dx^i dx^j,
\label{splitting}
\end{equation}
where $t$ is the time coordinate given by projecting from $M=\mathbb{R} \times \Sigma$ onto the first factor and $x=(x^1,x^2,x^3)$ are coordinates on $\Sigma$. 

For the following we have to assume that we have a globally hyperbolic spacetime with an orthogonal splitting that satisfies the so-called \emph{metric growth condition} which was introduced by Uhlenbeck \cite{UHLENBECK1975}. By definition, this condition is fulfilled if and only if for every compact subset of $\Sigma$ there is a function $F$ that satisfies
\begin{equation}
\int_{-\infty}^0 \dfrac{dt}{F(t)}=\infty
\label{3}
\end{equation}
such that for $t \le 0$ the inequality 
\begin{equation}
H_{ij}(x,t)v^i v^j \le f(x,t)F(t)^2 G_{ij}(x)v^i v^j
\label{eq:metricgrowth}
\end{equation}
holds for all $x$ in the compact subset and for all $(v^1, v^2, v^3) \in \mathbb{R}^3$, with a time-independent Riemannian metric $G_{ij}$ on $\Sigma$.

Note that we have given here Uhlenbeck's metric growth condition in a time-reversed way; the reason is that Uhlenbeck wanted to determine future-oriented lightlike geodesics from a point $p$ to a timelike curve $\gamma$ whereas for applications to lensing we are interested in past-oriented lightlike geodesics from $p$ to $\gamma$. In this time-reversed version, the metric growth condition prohibits the existence of particle horizons, i.e., it guarantees that from each point $p$ in $M$ one can find a past-pointing lightlike curve to every timelike curve that is vertical with respect to the orthogonal splitting chosen.

For the following we also need the notions of conjugate points and of Betti numbers. 

Recall that a point $q$ is said to be \emph{conjugate} to a point $p$ along a geodesic $\lambda$ if there exists a non-zero Jacobi field (i.e., a non-trivial geodesic variation) along $\lambda$ which vanishes at $p$ and $q$. Such Jacobi fields form a vector space and the dimension of this vector space is called the \emph{multiplicity} of the conjugate point. For lightlike geodesics, multiples of the tangent field have to be factored out of the space of Jacobi fields. The (Morse) \emph{index} of a given geodesic that starts at $p$ is the number of points $q$ conjugate to $p$, counting multiplicities.  The set of all points that are conjugate to a given point $p$, along any (past-pointing) lightlike geodesic, is called the (past) \emph{caustic} of $p$.  

Given a topological manifold $\mathcal{M}$, the $\kappa$th \emph{Betti number} $B_{\kappa}$ of $\mathcal{M}$ is the dimension of the $\kappa$th homology space of $\mathcal{M}$ with coefficients in $\mathbb{R}$. Descriptively, $B_0$ counts the connected components of $\mathcal{M}$ and $B_{\kappa}$, for $\kappa>0$, counts those holes in $\mathcal{M}$ that prevent a $\kappa$-dimensional sphere from being a boundary. In the case we are interested in, $\mathcal{M}$ will be the loop space $L(M)$ of the spacetime manifold $M$, i.e., the set of all continuous maps from the circle $S^1$ to $M$ that go through a fixed point $p$ in $M$. We assume, of course, that the spacetime manifold $M$ is connected; then $L(M)$ is independent of which point $p$ we have chosen.  

We are now ready for stating Uhlenbeck's theorem. 
\begin{theorem}
Let $(M,g)$ be a globally hyperbolic spacetime  that admits an orthogonal splitting \eqref{s1} and \eqref{splitting} which satisfies the metric growth condition. Fix a point $p \in M$ and a smooth timelike curve $\gamma: \mathbb{R}\to M$ that takes the form $\gamma(\tau)=(\beta(\tau),\tau)$ with respect to the chosen splitting. Assume that $\gamma$ does not meet the caustic of the past light-cone of $p$ and that for some sequence $(\tau_i)_{i\in \mathbb{N}}$ with $\tau_i \to -\infty$ the sequence $(\beta(\tau_i))_{i \in \mathbb{N}}$ converges in $\Sigma$. Then the Morse inequalities
\begin{equation}
N_{\kappa} \ge B_{\kappa} \;\text{for all}\; {\kappa} \in \mathbb{N}_0
\label{5}
\end{equation}
and the Morse relation 
\begin{equation}
\sum_{{\kappa}=0}^{\infty}(-1)^{\kappa} N_{\kappa}=
\sum_{{\kappa}=0}^{\infty}(-1)^{\kappa} B_{\kappa}
\label{6}
\end{equation}
hold, 
where $N_{\kappa}$ denotes the number of past-pointing lightlike geodesics with index ${\kappa}$ from $p$ to $\gamma$, and $B_{\kappa}$ denotes the ${\kappa}$th Betti number of the loop space of $M$.
\label{theorem}
\end{theorem}
\begin{proof}
See Uhlenbeck \cite{UHLENBECK1975}, Sec. IV and Proposition 5.2.
\hfill
$\blacksquare$
\end{proof}

We should point out that the convergence condition on $(\beta(\tau_i))_{i\in \mathbb{N}}$ is certainly satisfied if $\beta$ is confined to a compact subset of $\Sigma$, i.e., if $\gamma$ stays in a spatially compact set.

The rotating traversable wormhole spacetime that we will consider below has topology $S^2 \times \mathbb{R}^2$ where $S^2$ is the 2-sphere. As this space is simply connected but not contractible to a point, a theorem by Serre \cite{Serre1951} implies that for all but finitely many $\kappa \in \mathbb{N}_0$ we have $B_{\kappa}>0$. With this condition, \eqref{5} says that $N_{\kappa}>0$ for all but finitely many $\kappa$, i.e for almost every positive integer $\kappa$ we can find a past-pointing lightlike geodesic from $p$ to $\gamma$ with $\kappa$ conjugate points in its interior. Hence, there must be infinitely many past-pointing lightlike geodesics from $p$ to $\gamma$.

\section{ROTATING TRAVERSABLE WORMHOLES}\label{sec:worm}
According to Teo \cite{Teo:1998dp}, a stationary and axisymmetric metric 
suitable for describing a rotating traversable wormhole is given by
\[
g= -\tilde{N}(r,\vartheta)^2dt^2+\Big(1-\dfrac{b(r,\vartheta)}{r} \Big)^{-1}dr^2
\]
\begin{equation}
+\tilde{R}(r,\vartheta)^2\bigg[d\vartheta^2+\sin^2\vartheta(d\varphi-\omega(r,\vartheta)dt)^2 \bigg] .
\label{metric1}
\end{equation}
Here the time coordinate $t$ runs over all of $\mathbb{R}$, the radial coordinate $r$ is restricted by the condition $b(r, \vartheta) < r < \infty$, and $\vartheta$ and $\varphi$ have their usual range as spherical coordinates. On this domain, the metric functions $\tilde{N}$, $\tilde{R}$ and $b$ are assumed to be strictly positive and the condition of asymptotic flatness is assumed to hold, $\tilde{N} (r , \vartheta )\to 1$, $\tilde{R} (r, \vartheta )/r \to 1$, $b(r , \vartheta ) /r \to 0$ and $ r \omega ( r , \vartheta ) \to 0$ for $r \to \infty$. The metric functions have the following meaning. $\tilde{N}$ is the so-called lapse function that relates the time coordinate $t$ to proper time along the $t$-lines. $\tilde{R}$ determines the proper circumference, $2\pi \tilde{R} (r , \vartheta ) \sin \vartheta$, of the circle located at the coordinate values $(r,\vartheta)$, with $\varphi$ ranging from 0 to $2\pi$. $\omega$ determines the twist of the $t$-lines, i.e., the rotation of the wormhole, and $b$ determines the location of the throat. For the latter to be well-defined and regular, one has to require that $\partial _{\vartheta} b \to 0$ and $b-r\partial_r b \to 0$ for $b/r \to 1$. The first condition makes sure that the equation $b(r , \vartheta) = r$ determines a unique radius value, $r_0$, that is independent of $\vartheta$; if this condition is violated, the Ricci scalar of the metric diverges to infinity at the throat, see Teo \cite{Teo:1998dp}, i.e., it is not possible to analytically extend the metric beyond the throat. The second condition is known as the ``flare-out condition''; it makes sure that, if the first condition is satisfied, the area of the sphere at $r=\mathrm{constant}$ approaches a local minimum for $r \to r_0$. If both conditions are satisfied we may join two copies of the metric, each with the coordinate $r$ running from $r_0$ to $\infty$, at the throat, thereby getting a wormhole spacetime with two asymptotically flat ends. If the metric functions $\tilde{N}$, $\tilde{R}$ and $b$ are independent of $\vartheta$ and if $\omega =0$ one gets the spherically symmetric and static class of wormholes discussed by Morris and Thorne \cite{MorrisThorne1988}. 

Teo's representation (\ref{metric1}) of rotating wormholes is convenient for considerations restricted to the region between one asymptotic end and the throat. Here, however, we want to use methods for which the global topology of the manifold is relevant. For such considerations it is desirable to introduce a radial coordinate that covers the entire spacetime. As pointed out already by Teo \cite{Teo:1998dp}, this is possible if $b=b(r)$ is independent of $\vartheta$. Then we may define a new radial coordinate, $\ell$, by
\begin{equation}
\dfrac{d \ell }{dr}=\pm \Big(1-\dfrac{b(r)}{r}\Big)^{-1/2} \, .
\label{r}
\end{equation}
The metric \eqref{metric1} then becomes 
\[
g= -\tilde{N}(\ell,\vartheta)^2dt^2+ d\ell^2
\]
\begin{equation}
+\tilde{R}(\ell,\vartheta)^2\bigg[d\vartheta^2+\sin^2\vartheta(d\varphi-\omega(\ell,\vartheta)dt)^2 \bigg] \, .
\label{m2}
\end{equation}
Here $\tilde{N}$, $\tilde{R}$ and $\omega$ are the same quantities as before, but now with $r$ replaced by the new coordinate $\ell$ and analytically extended to the range $\ell \in \, ]-\infty,+\infty[$. In this way the entire wormhole spacetime, from one asymptotic end at $\ell = - \infty$ to the other one at $\ell = + \infty$, is covered by a single coordinate system (which features the usual coordinate singularities of the angular coordinates). As we read from the metric, $\ell$ gives proper length along each radial line.

Note that it is possible to generalise the component $g_{\ell \ell}=1$ in the metric \eqref{m2} to $g_{\ell \ell}=h(\ell,\vartheta)^2$. If $h ( \ell, \vartheta )$ is strictly positive on the entire spacetime and approaches 1 for $\ell \to \pm \infty$, this modification does not violate the regularity or the asymptotic flatness. The metric \eqref{m2} then becomes
\[
g= -\tilde{N}(\ell,\vartheta)^2dt^2+h(\ell,\vartheta)^2 d\ell^2
\]
\begin{equation}
+\tilde{R}(\ell,\vartheta)^2\bigg[d\vartheta^2+\sin^2\vartheta(d\varphi-\omega(\ell,\vartheta)dt)^2 \bigg]\\
\label{h}
\end{equation} 
This metric describes a class of spacetimes that contains all Teo wormholes with $b = b(r)$. On the other hand, it also includes metrics which are not of the Teo type: If we require, as the only conditions on the  metric coefficients in (\ref{h}), that $\tilde{N}$, $\tilde{R}$ and $h$ are strictly positive and that the condition of asymptotic flatness is satisfied for both $\ell \to - \infty$ and $\ell \to + \infty$, we get a class of mathematical models that describe wormholes in the sense that we have a spacetime without singularities or horizons that connects two asymptotically flat ends. For each $\vartheta$, we can determine the circumference of the circle at $(\ell, \vartheta )$ as a function of $\ell$. As this circumference goes to infinity for $\ell \to - \infty$ and for $\ell \to + \infty$, this function must have at least one local minimum. However, in contrast to the Teo wormholes, the location of this minimum may depend on $\vartheta$, and there may be several local minima (``throats'') with local maxima (``bellies'') in between. In general, there is no symmetry with respect to reflections $\ell \mapsto - \ell$ and also not with respect to reflections $\vartheta \mapsto \pi - \vartheta$. It is the class of wormholes given by (\ref{h}) to which we want to apply Morse theory in this paper.

As the lightlike geodesics and the conformal structure of a spacetime remain unchanged if we perform a conformal transformation with the conformal factor $h( \ell  , \vartheta )^2$, we may switch from the metric (\ref{h}) to the conformally equivalent metric 
\[
g=-\dfrac{\tilde{N}(\ell,\vartheta)^2}{h(\ell,\vartheta)^2}dt^2+d\ell^2
\]
\begin{equation}
+\dfrac{\tilde{R}(\ell,\vartheta)^2}{h(\ell,\vartheta)^2}\bigg[d\vartheta^2+\sin^2\vartheta(d\varphi-\omega(\ell,\vartheta)dt)^2 \bigg] \, . \\
\label{metricq}
\end{equation}
If we define 
\begin{equation}
N:=\dfrac{\tilde{N}}{h} \, , \quad 
R:=\dfrac{\tilde{R}}{h} \, ,
\label{eq:NR}
\end{equation}
we get the metric
\[
g= -N(\ell,\vartheta)^2dt^2+d\ell^2
\]
\begin{equation}
+R(\ell,\vartheta)^2\bigg[d\vartheta^2+\sin^2\vartheta(d\varphi-\omega(\ell,\vartheta)dt)^2 \bigg]\\
\label{metric}
\end{equation}
which is the same as (\ref{m2}). Note that with (\ref{eq:NR}) the regularity and the asymptotic flatness of the metric (\ref{h}) guarantees the regularity and the asymptotic flatness of the metric (\ref{metric}). Therefore, for the rest of this paper, we consider a spacetime $(M,g)$, where $g$ is the metric \eqref{metric}, with the requirement that $N$ and $R$ are strictly positive and that for $\ell \to \pm \infty$
\begin{equation}
N =  1 + O \big( 1/|\ell | \big) \, , \quad
R = | \ell | \big( 1 + O ( 1/| \ell | \big) \, , \quad
\omega = O \big( 1/ | \ell | ^2 \big) \, .
\label{asy}
\end{equation}

Notice that in spacetimes with the metric \eqref{metricq} or \eqref{metric} the Hamilton-Jacobi equation for lightlike geodesics is not in general separable. There are of course special cases where a generalised Carter constant exists which allows to separate the Hamilton-Jacobi equation. This is true, in particular, if in the metric (\ref{metric}) the functions $N$, $R$ and $\omega$ are independent of $\vartheta$. However, for the purpose of this paper it is not necessary to restrict to such cases. It is one of the major advantages of the methods to be applied in this paper that they do not require the existence of a generalised Carter constant. 

Finally we mention that, in general, there exists an ergoregion in our wormhole spacetimes, i.e. a region where $g_{tt} = R^2 \omega ^2-N^2 >0$. However, because of the asymptotic flatness, the ergoregion cannot extend to infinity, i.e., it is restricted to a spatially compact domain $|\ell |< \ell_{max}$. The ergoregion, if it exists, will be of no particular relevance for the following discussion.

\subsection{GLOBAL HYPERBOLICITY AND METRIC GROWTH CONDITION OF THE WORMHOLE METRIC}
For applying Uhlenbeck's theorem to our wormhole spacetimes we first have to demonstrate that the latter are globally hyperbolic and satisfy the metric growth condition. To that end, we use some known results on stationary spacetimes, i.e., on  spacetimes $(M,g)$ where $M$ is a product manifold, $M = \mathbb{R} \times \Sigma$, of the real line $\mathbb{R}$ and a 3-dimensional manifold $\Sigma$, and the metric is of the form
\begin{equation}
g=-N(x) ^2 dt^2+g_{ij}(x)(dx^i+\beta^i (x)  dt)(dx^j+\beta^j (x) dt) \, ,
\end{equation}
where $x = (x^1,x^2,x^3)$ are coordinates on $\Sigma$. Clearly, our wormhole spacetimes are of this form, where $\Sigma = \mathbb{R} \times S^2$ with $(x^1,x^2,x^3) = (\ell, \vartheta, \varphi)$. From (\ref{metric}) we read the ``lapse function'' $N(x)=N(\ell , \vartheta )$, the ``shift vector'' $\beta ^i (x) \partial _i = - \omega ( \ell , \vartheta ) \partial _{\varphi}$ and the spatial metric $g_{ij} (x) dx^idx^j = d \ell ^2 + R(\ell , \vartheta )^2 (d \vartheta ^2 + \mathrm{sin} ^2 \vartheta \, d \varphi ^2 )$.

The asymptotic properties (\ref{asy}) of the wormhole metric are crucial for proving the following.
\begin{proposition}
The wormhole spacetime $(M,g)$ is globally hyperbolic.
\label{prop:ghyp}
\end{proposition}
\begin{proof}
We establish three properties. (i) There are positive constants $N_1$ and $N_2$ such that the lapse function $N(\ell , \vartheta )$ satisfies $0 < N_1< N( \ell , \vartheta ) < N_2$ on the entire spacetime. This follows immediately from the facts that the lapse function is everywhere strictly positive and that, by (\ref{asy}), it goes to 1 for $\ell \to \pm \infty$. (ii) The spatial part of the wormhole spacetime, $( \Sigma , g_{ij}(x) dx^idx^j)$, is a complete Riemannian manifold. To prove this, we observe that the coordinate function $\ell$ can be viewed as a function $\ell: \Sigma = \mathbb{R} \times S^2 \to \mathbb{R}$, defined just by projecting onto the first factor. Obviously, this function is proper, i.e., for every compact subset $I \subset \mathbb{R}$ the pre-image $\ell ^{-1} (I)$ is compact. Moreover, the gradient of this function has constant norm 1 with respect to the metric $g_{ij} (x) dx^idx^j$. We have thus proven that the Riemannian manifold $( \Sigma , g_{ij}(x) dx^idx^j)$ admits a proper function with bounded norm. According to a general result by Gordon \cite{10.2307/2038738}, this implies that this Riemannian manifold is complete. (iii) There is a positive constant $B$ that bounds the norm of the shift vector $\beta ^i (x) \partial _i = - \omega ( \ell , \vartheta )\partial _{\varphi}$, i.e., $g_{ij} (x) \beta ^i (x) \beta ^j (x) = R( \ell , \vartheta ) ^2 \mathrm{sin} ^2 \vartheta \, \omega ( \ell , \vartheta )^2 \le B^2$. This follows from the facts that, by (\ref{asy}), the function $R( \ell , \vartheta ) ^2 \mathrm{sin} ^2 \vartheta \, \omega ( \ell , \vartheta ) ^2$ goes to zero for $\ell \to \pm \infty$ and that this function has no singularities. Having established the three properties (i), (ii) and (iii), we can now refer to a result by Choquet-Bruhat and Cotsakis \cite{ChoquetBruhat:2002su} who have shown that these three properties imply that the spacetime metric is globally hyperbolic. Note that Choquet-Bruhat and Cotsakis allow the spatial metric to be time-dependent. Then one also has to establish that it is bounded below by a time-independent metric. As our spatial metric is time-independent, this condition is trivially satisfied. \hfill $\blacksquare$
\end{proof}

\begin{proposition}
The wormhole spacetime admits an orthogonal splitting that
satisfies the metric growth condition.
\label{prop:growth}
\end{proposition}
\begin{proof}
In the metric \eqref{metric} the $t$-lines are not orthogonal to the surfaces 
$t =$  constant. Therefore, we change to new spatial coordinates
\begin{equation}
x^1=\ell, \quad 
x^2=\vartheta,\quad
x^3=\varphi- \omega (\ell,\vartheta)t.
\end{equation}
Then the metric \eqref{metric} of the rotating traversable wormhole takes the 
orthogonal splitting form \eqref{splitting}, with
\[
H_{ij}(x,t)dx^i dx^j=d\ell^2+R ( \ell , \vartheta ) ^2 d\vartheta^2
\]
\begin{equation}
+R ( \ell , \vartheta ) ^2 \sin\vartheta^2 \Big(t(\dfrac{\partial \omega(\ell,\vartheta)}{\partial \ell}d\ell+\dfrac{\partial \omega(\ell,\vartheta)}{\partial \vartheta}d\vartheta)+dx^3\Big)^2
\label{eq:split1}
\end{equation}
and
\begin{equation}
f(x,t)=N(\ell,\vartheta)^2 \, .
\label{eq:split2}
\end{equation}
If we restrict the range of the coordinates $x=(x^1,x^2,x^3)$ to a compact set in $\Sigma= \mathbb{R} \times S^2$, we read from \eqref{eq:split1} and \eqref{eq:split2} that there are positive constants $A$ and $B$ such that
\begin{equation}
\dfrac{H_{ij}(x,t)v^i v^j}{f(x,t)} \le 
(A+B|t|)^2 \delta_{ij}v^i v^j
\label{spw}
\end{equation}
for all $(v^1,v^2,v^3) \in \mathbb{R}^3$. This demonstrates that the metric growth condition \eqref{eq:metricgrowth} holds, with $F(t)= A+B|t|$ and $G_{ij}= \delta _{ij}$. 
\hfill
$\blacksquare$
\end{proof}

\section{INERTIAL FORCES IN THE WORMHOLE SPACETIME}\label{sec:inertial}
We derive now the inertial forces for observers on circular orbits around the axis of rotational symmetry in the wormhole spacetime $(M,g)$ with the metric \eqref{metric}. This will allow us to define two potentials $\Psi _{\pm}$ that give us important information on lightlike geodesics. For our discussion it will be helpful to introduce the following orthonormal basis on the spacetime $(M,g)$:
\begin{equation}
E_0=\dfrac{1}{N} \big( \partial_t+ \omega \, \partial_{\varphi} \big) \, , \quad
E_1=\partial_\ell \, , \quad
E_2=\dfrac{1}{R}\partial_{\vartheta} \, , \quad
E_3= \dfrac{1}{R\,\sin\vartheta}\partial_{\phi} \, , \quad
\end{equation}
whose dual basis is given by the covector fields
\begin{equation}
-g(E_0,.)=N\,dt \, , \:
g(E_1,.)=d\ell \, , \:
g(E_2,.)=R\,d\vartheta \, , \:
g(E_3,.)= R \, \sin\vartheta  \big( d \varphi - \omega \, dt \big) \, .
\end{equation}
For later calculations we list all nonvanishing Lie brackets of the $E_{\mu}$,
\begin{equation}
[E_0,E_1]=\dfrac{\partial_{\ell}N}{N}\;E_0-\dfrac{R \partial_{\ell}\omega\sin\vartheta}{N}\;E_3 \, ,
\end{equation}
\begin{equation}
[E_0,E_2]=\dfrac{\partial_{\vartheta}N}{R N}\;E_0-\dfrac{\partial_{\vartheta}\omega\sin\vartheta}{N} \;E_3 \, ,
\end{equation}
\begin{equation}
[E_1,E_2]=-\dfrac{\partial_{\ell}R}{R}\;E_2 \, ,
\end{equation}
\begin{equation}
[E_1,E_3]=-\dfrac{\partial_{\ell}R}{R}\;E_3 \, ,
\end{equation}
\begin{equation}
[E_2,E_3]=-\dfrac{\sin\vartheta\;\partial_{\vartheta}R+\cos\vartheta\;R}{R^2\sin\vartheta}\;E_3 \, .
\end{equation}
The 4-velocities of observers who circle along the $\varphi$-lines are given by
\begin{equation}
U=\gamma \big( E_0 \pm  v \, E_3  \big) \quad \text{with} \quad 
\gamma:=\dfrac{1}{\sqrt{1-v^2}}
\label{U}
\end{equation}
where the number $v\in [0,1]$ gives the velocity (in units of the velocity of light) of these observers with respect to the stationary observers whose worldlines are the $t$-lines. Note that the integral curves of $U$ are parametrised by proper time, i.e., $g(U,U)=-1$. For the upper sign in \eqref{U}, the motion relative to the stationary observers is in the positive $\varphi$-direction, for the negative sign it is in the negative $\varphi$-direction. Clearly, $U$ is non-geodesic, $\nabla_U U 
\neq 0$, i.e., one needs a thrust to stay on an integral curve of $U$.

With $g(\nabla_U U,E_{\mu})=-g(U,[U,E_{\mu}])$, the tetrad components of $\nabla _U U$ are determined by the Lie brackets that we have calculated above. Thereupon the acceleration  of a freely falling particle relative to the $U$-observer can be decomposed into three parts, 
\begin{equation}
-g(\nabla_U U, \, . \, )= A_{\mathrm{grav}}+A_{\mathrm{cor}}+A_{\mathrm{cent}} \, ,
\end{equation}
according to the rule that the gravitational acceleration is independent of $v$,
\begin{equation}
A_{\mathrm{grav}}= -\dfrac{\partial_{\ell}N}{N}\;d\ell-\dfrac{\partial_{\vartheta}N}{N}\;d\vartheta,
\end{equation} 
the Coriolis acceleration is odd with respect to $v$,
\[
A_{\mathrm{cor}}=\pm\dfrac{v}{(1-v^2)}\Bigg(-\dfrac{R\; \partial_{\ell}\omega\;\sin\vartheta}{N}\;d\ell 
\]
\begin{equation}
-\dfrac{R\;\partial_{\vartheta}\omega\;\sin\vartheta}{N}\;d\vartheta\Bigg),
\end{equation}
and the centrifugal acceleration  is even with respect to $v$,
\[
A_{\mathrm{cent}}=\dfrac{v^2}{(1-v^2)}\Bigg(\Big(-\dfrac{\partial_{\ell}N}{N}+
\dfrac{\partial_{\ell}R}{R}\Big)\;d\ell
\]
\begin{equation}
+\Big(-\dfrac{\partial_{\vartheta}N}{N}+\dfrac{\sin\vartheta\;\partial_{\vartheta}R+R\cos\vartheta}{R\sin\vartheta}\Big)d\vartheta\Bigg) \, .
\end{equation}
Multiplying the inertial acceleration with the rest mass of the freely falling particle gives the corresponding inertial force. 

Quite generally, the gravitational, Coriolis and centrifugal accelerations are unambiguously defined whenever a timelike 2-surface with a timelike vector field has been specified, see Foertsch et al. \cite{FoertschEtAl2003}. Here we apply this procedure to each 2-surface $(\ell,\vartheta)= \,$constant with the timelike vector field $E_0$.

We want to investigate the behaviour of the inertial accelerations if $v$ approaches the velocity of light. If we take the sum of Coriolis and centrifugal acceleration up to the positive factor $v/(1-v^2)$, we find: 
\[
Z_{\pm}(v)=\pm \Bigg(-\dfrac{R\;\partial_{\ell}\omega\sin\vartheta}{N}\;d\ell
-\dfrac{R\;\partial_{\vartheta}\omega\sin\vartheta}{N}\;d\vartheta\Bigg)
\]
\begin{equation}
+
v\Bigg(\Big(-\dfrac{\partial_{\ell}N}{N}+\dfrac{\partial_{\ell}R}{R}\Big)\;d\ell+\Big(-\dfrac{\partial_{\vartheta}N}{N}+\dfrac{\sin\vartheta\;\partial_{\vartheta}R+R \cos\vartheta}{R\sin\vartheta}\Big)d\vartheta\Bigg)
\end{equation}
To consider the behavior for $v$ close to the velocity of light, we take the limit $v\to 1$, 
\[
\displaystyle{\lim_{v \to 1}}Z_{\pm}= \Bigg(-\dfrac{\partial_{\ell}N}{N}+\dfrac{\partial_{\ell}R}{R}\mp \dfrac{R\;\partial_{\ell}\omega \sin\vartheta }{N} \Bigg)\;d\ell
\]
\begin{equation}
+\Bigg(-\dfrac{\partial_{\vartheta}N}{N}+\dfrac{\partial_{\vartheta}R}{R}+\dfrac{\cos\vartheta}{\sin\vartheta}\mp \dfrac{R\partial_{\vartheta}\omega\sin\vartheta}{N} \Bigg)d\vartheta
\end{equation}
This can be rewritten as 
\begin{equation}
\displaystyle{\lim_{v \to 1}}Z_{\pm}=  \dfrac{R \sin\vartheta}{N}\; d\Psi_{\pm} 
\label{limit}
\end{equation}
where 
\[
d\Psi_{\pm}=\Big(-\dfrac{\partial_{\ell}N}{R\sin\vartheta}+\dfrac{N \partial_{\ell}R}{R^2 \sin\vartheta}\mp \partial_{\ell}\omega \Big)\;d\ell
\]
\begin{equation}
+\Big(-\dfrac{\partial_{\vartheta}N}{R \sin\vartheta}+\dfrac{N \partial_{\vartheta}R}{R^2 \sin\vartheta}+\dfrac{N \cos\vartheta}{R\sin^2\vartheta}\mp \partial_{\vartheta}\omega \Big) \; d\vartheta
\label{dpsi}
\end{equation}
is the differential of the function
\begin{equation}
\Psi_{\pm}=-\dfrac{N}{R \sin\vartheta}\mp \omega \, .
\label{pot}
\end{equation}
As in \eqref{pot} there is a factor of $\sin\vartheta$ in the denominator, both $\Psi_{+}$ and $\Psi_{-}$ are singular along the axis. Outside the ergoregion $\Psi_{+}$ is negative and $\Psi_{-}$ is positive, and inside the ergoregion (if there is any) one of the two potentials changes sign.

From the asymptotic flatness it follows that 
\begin{align}
\Psi _{\pm} = - \dfrac{1}{ \ell  \, \mathrm{sin} \, \vartheta} \Big( 1 + O \big( 1/ | \ell |) \Big)     
\quad \text{for} \: \ell \to \infty \, , 
\\
\Psi _{\pm} =  \dfrac{1}{ \ell  \, \mathrm{sin} \, \vartheta} \Big( 1 + O \big( 1/ | \ell |) \Big)     
\quad \text{for} \: \ell \to - \infty \, ,
\end{align}
and
\begin{align}
\partial _{\ell} \Psi _{\pm} =  \dfrac{1}{ \ell ^2 \, \mathrm{sin} \, \vartheta} \Big( 1 + O \big( 1/ | \ell |) \Big)     \quad \text{for} \: \ell \to \infty \, , 
\label{eq:asyPsi1}
\\
\partial _{\ell} \Psi _{\pm} =  - \dfrac{1}{ \ell ^2 \, \mathrm{sin} \, \vartheta} \Big( 1 + O \big( 1/ | \ell |) \Big)     \quad \text{for} \: \ell \to - \infty \, . 
\label{eq:asyPsi2}
\end{align}

Eq. \eqref{limit} tells us that, in the limit $v\to 1$, the sum of Coriolis and centrifugal force is perpendicular to the surfaces $\Psi_{\pm}=\,$constant and points in the direction of increasing $\Psi_{\pm}$. In this limit, we may thus view the function $\Psi_{+}$ (or $\Psi_{-}$, respectively) as a Coriolis-plus-centrifugal potential for co-rotating (or counter-rotating, respectively) observers. The surfaces $\Psi_{\pm}=\,$constant are shown for example spacetimes in Figures \ref{psi} and \ref{psi2}.

The potentials $\Psi_{\pm}$ are quite analogous to the potentials that were introduced by Hasse and Perlick \cite{Hasse:2005vu} for the Kerr-Newman metric. We will see that these potentials are relevant for lensing because they tell us where the radius coordinate $\ell$ has minima or maxima along a lightlike geodesic.

With the help of the potentials  $\Psi_{\pm}$, we decompose the wormhole spacetime in the following way: 
\begin{definition}
We define the regions $M_{out}$, $M_{in}$, $K_{+}$ and $K_{-}$ by the following properties:
\begin{equation}
\partial_{\ell}\Psi_{+}<0\;\text{and}\;\partial_{\ell}\Psi_{-}<0 \;\text{on}\;M_{in}, 
\end{equation}
\begin{equation}
\partial_{\ell}\Psi_{+}<0\;\text{and}\;\partial_{\ell}\Psi_{-}>0 \;\text{on}\;K_{-}, 
\end{equation}
\begin{equation}
\partial_{\ell}\Psi_{+}>0\;\text{and}\;\partial_{\ell}\Psi_{-}<0 \;\text{on}\;K_{+}, 
\end{equation}
\begin{equation}
\partial_{\ell}\Psi_{+}>0\;\text{and}\;\partial_{\ell}\Psi_{-}>0 \;\text{on}\;M_{out}. 
\end{equation}
We also define the closed set $K = M \setminus \big( M_{in} \cup M_{out} \big)$
\label{def:inout}
\end{definition} 

The following proposition follows from this definition. 
\begin{proposition} 
\begin{itemize}
\item[(a)]
$M_{\mathrm{out}}$ is the set of all events where
\begin{equation}
   - \dfrac{\partial_{\ell}N}{N}+\dfrac{\partial_{\ell}R}{R}
   > \Big|  \dfrac{R \, \mathrm{sin} \, \vartheta}{N} \, \partial_{\ell}\omega \Big|  
\end{equation}
and $M_{\mathrm{in}}$ is the set of all events where
\begin{equation}
   - \dfrac{\partial_{\ell}N}{N}+\dfrac{\partial_{\ell}R}{R}
   < - \Big| \dfrac{R \, \mathrm{sin} \, \vartheta}{N} \, \partial_{\ell}\omega \Big|  
\end{equation}
\item[(b)]
There are $\ell _1$ and $\ell _2$ such that the region $- \infty < \ell < \ell _1$ is completely 
contained in $M_{\mathrm{in}}$ and the region $\ell _2 < \ell < \infty$ is completely contained in $M_{\mathrm{out}}$.
\end{itemize}
\label{prop:inout}
\end{proposition}
\begin{proof}
From \eqref{dpsi} we read that 
\begin{equation}
\partial _{\ell} \Psi _{\pm} =
\dfrac{N}{R \, \mathrm{sin} \, \vartheta}
\Big( - \dfrac{\partial_{\ell}N}{N}+\dfrac{\partial_{\ell}R}{R}\mp \dfrac{R \, \mathrm{sin} \, \vartheta}{N} \, \partial_{\ell}\omega \Big)
\end{equation}
which implies part (a). Part (b) follows immediately from \eqref{eq:asyPsi1} and \eqref{eq:asyPsi2}. 
\hfill $\blacksquare$
\end{proof}

It can be read from the definitions that, for $v$ sufficiently close to 1, in $M_{\mathrm{in}}$ the direction of centrifugal-plus-Coriolis force is always pointing in the direction of decreasing $\ell$ and in $M_{\mathrm{out}}$ it is always pointing in the direction of increasing $\ell$. This means that in these regions the centrifugal-plus-Coriolis force is always pointing away from the centre, for co-rotating and counter-rotating observers, which is the situation one is used to from Newtonian physics. By contrast, in the interior of the regions $K$ the centrifugal-plus-Coriolis force points in the reverse direction, either for co-rotating or for counter-rotating observers. Therefore, the boundary of the region $K$ determines the points where centrifugal-plus-Coriolis force reversal takes place.

Moreover, the regions $K$, $M_{in}$ and $M_{out}$ are helpful because they tell us where the radius coordinate $\ell$ has minima and maxima along a lightlike geodesic. We present these properties in the following proposition. 

\begin{proposition} \label{proposition2}
a) In the region $M_{out}$, the radius coordinate $\ell$ cannot have other extrema than
strict local minima along a lightlike geodesic.\\
b) In the region $M_{in}$, the radius coordinate $\ell$ cannot have other extrema than
strict local maxima along a lightlike geodesic.\\
c) 
Through each point of $K$ there is a lightlike geodesic such that the first and second derivatives of $\ell$ with respect to the affine parameter vanish at this point. 
\label{prop:minmax}
\end{proposition}
\begin{proof}
Let $X$ be a lightlike and geodesic vector field on $(M,g)$ i.e $g(X,X)=0$ and $\nabla_X X=0$. For proving 
(a) and (b) we have to demonstrate that the implication
\begin{equation}
X\ell=0 \implies XX\ell>0
\label{eq:posMin}
\end{equation}
is true at all points of $M_{\mathrm{out}}$ and that the implication
\begin{equation}
X\ell=0 \implies XX\ell<0
\label{eq:negMout}
\end{equation}
is true at all points of $M_{\mathrm{in}}$. 
The equation  $\nabla_X X=0$ implies 
\begin{equation}
XX \ell= X d\ell(X)= X(g(E_1,.))=g(\nabla_X E_1,.)
\end{equation}
where 
\begin{equation}
X=E_0+\cos\alpha E_3+\sin\alpha E_2
\end{equation}
then 
\begin{equation}
\begin{aligned}
XX \ell&=g(\nabla_{E_0} E_1,E_0)+\sin\alpha (g(\nabla_{E_2} E_1,E_0)+g(\nabla_{E_0} E_1,E_2))\\
&+\cos\alpha(g(\nabla_{E_3} E_1,E_0)+g(\nabla_{E_0} E_1,E_3))
+\sin^2\alpha g(\nabla_{E_2} E_1,E_2)\\
&+\cos^2\alpha g(\nabla_{E_3} E_1,E_3)\\
&=g([E_0,E_1],E_0)+\sin\alpha(g([E_2,E_1],E_0)+g([E_0,E_1],E_2))\\
&+\cos\alpha(g([E_3,E_1],E_0)+g([E_0,E_1],E_3))+\sin^2\alpha g([E_2,E_1],E_2)\\
&+\cos^2\alpha g([E_3,E_1],E_3)
\end{aligned}
\end{equation}
with the use of the Lie brackets, we find 
\begin{equation}
XX\ell=\dfrac{\partial_{\ell}R}{R}-\dfrac{\partial_{\ell}N}{N}-\cos\alpha \, \Big(\dfrac{R\;\partial_{\ell}\omega\sin\vartheta}{N} \Big)
\label{pro1}
\end{equation}
If $\cos\alpha$ runs through all possible values from $-1$ to $1$, the right-hand side of \eqref{pro1} stays positive on $M_{\mathrm{out}}$ and negative on $M_{\mathrm{in}}$, by Proposition \ref{prop:inout}. This proves part (a) and part (b). At every point of $K$ there is a value of $\cos\alpha$ such that the right-hand side of \eqref{pro1} vanishes. This proves part (c) of the proposition.  
\hfill $\blacksquare$
\end{proof}

We have already emphasised that in our wormhole spacetimes the Hamilton-Jacobi equation for lightlike geodesics is \emph{not} in general separable. In the special case that it is separable, i.e., in the case that a generalised Carter constant exists, part (c) of Proposition \ref{proposition2} implies that through each point of the region $K$ there is a \emph{spherical} lightlike geodesic, i.e., a lightlike geodesic that stays on a sphere $\ell=\,$constant. In this case one would call $K$ the \emph{photon region}. However, we do \emph{not} restrict to this special case in the following. 

\section{MULTIPLE IMAGING IN THE WORMHOLE SPACETIME}\label{sec:multiple}
In this part we want to apply Morse theory to get some information about the past-pointing lightlike geodesics from a point $p$ to a timelike curve $\gamma$ in the wormhole spacetime $(M,g)$. We first prove a proposition that characterises a region to which all lightlike geodesics between $p$ and $\gamma$ are confined. 
\begin{proposition}
Let $p$ be an event and $\gamma$ a past-inextendible timelike curve on which $|\ell |$ remains bounded if the time coordinate $t$ goes to $- \infty$ along $\gamma$. Let $\Lambda$ be the smallest shell $\ell _1 \leq \ell \leq \ell _2$ which contains $p$, $\gamma$ and the region $K$ defined in Definition \ref{def:inout}. Then every past-oriented  lightlike geodesic from $p$ to $\gamma$ is contained within $\Lambda$. 
\label{pro3}
\end{proposition}
\begin{proof}
By Proposition \ref{prop:inout}, along a lightlike geodesic that leaves and re-enters $\Lambda$ the radius coordinate $\ell$ must have either a maximum in the region $M_{out}$ or a minimum in the region $M_{in}$. Proposition \ref{proposition2} makes sure that this cannot happen.
\hfill $\blacksquare$
\end{proof}

We now use Uhlenbeck's theorem for determining the number of lightlike geodesics between $p$ and $\gamma$. As Morse theory applies only to functions for which the Hessian is non-degenerate at all critical points, we have to require that $\gamma$ does not meet the caustic of the past light-cone of $p$, i.e., that there is no past-pointing lightlike geodesic from $p$ which meets $\gamma$ in a point conjugate to $p$. 
\begin{proposition}
Consider, in the wormhole spacetime $(M,g)$, a point $p$ and a smooth future-pointing timelike curve $\gamma:]-\infty,\tau_a[ \to M$, with $-\infty<\tau_a \le \infty$, which is parametrised such that the t-coordinate of the point $\gamma(\tau)$ is equal to $\tau$. Assume 
\begin{itemize}
\item that $\gamma$ does not meet the caustic of the past light-cone of $p$, and
\item that for $\tau \to -\infty$ the radius coordinate $\ell$ of the point $\gamma(\tau)$ does not go to $-\infty$ or $+ \infty$.
\end{itemize}
Then there is an infinite sequence $(\lambda_n)_{n \in \mathbb{N} }$ of mutually different past-pointing lightlike geodesics from $p$ to $\gamma$. For $n \to \infty$, the index of $\lambda_n$ goes to infinity. Moreover, if we denote the point where $\lambda_n$ meets the curve $\gamma$ by $\gamma(\tau_n)$, then $\tau_n \to -\infty$ for $n \to \infty$.
\label{pro4}
\end{proposition}
\begin{proof}
In Proposition \ref{prop:ghyp} we have proven that the wormhole spacetime $(M,g)$ is globally hyperbolic and in Proposition \ref{prop:growth} we have shown that there is an orthogonal splitting, $M=\Sigma \times \mathbb{R}$ with $\Sigma \simeq S^2 \times \mathbb{R}$, such that the metric growth condition is satisfied. We now extend $\gamma$ to a curve that is defined for all time. More precisely, we choose a timelike curve $\gamma ' : \mathbb{R} \to M$ which takes the form $\gamma' (\tau ) = \big( \beta ' ( \tau ) , \tau \big)$ with respect to the orthogonal splitting such that $\gamma ' ( \tau ) = \gamma ( \tau )$ for $- \infty < \tau <  \tau _b$ with some $\tau _b \le \tau _a$. Our assumptions on $\gamma$ make sure that we can choose $\gamma '$ such that it does not meet the caustic of the past light-cone of $p$ and that $\big\{ \beta ' ( \tau ) \big| - \infty < \tau < \tau _b \big\}$ is confined to a compact subset of $\Sigma$. The latter property implies that for every sequence $ ( \tau _i ) _{i \in \mathbb{N}}$ the sequence $\big( \beta ' ( \tau _i ) \big) _{i \in \mathbb{N}}$ has a convergent subsequence. As a consequence, Uhlenbeck's theorem gives us the Morse inequalities $N_k' \ge B_k$, where $N_k'$ is the number of past-pointing lightlike geodesics from $p$ to $\gamma '$ with index $k$ and $B_k$ is the $k$th Betti number of the loop space of $M$. As $M \simeq S^2 \times \mathbb{R}^2$ is simply connected but not contractible to a point, a theorem by Serre \cite{Serre1951} implies that $B_k >0$ for all but finitely many $k$. We have thus proven that there is a past-pointing lightlike geodesic from $p$ to $\gamma '$ with index $k$ for all but finitely many $k$. In other words, there is an infinite sequence of mutually different past-pointing lightlike geodesics $(\lambda _n) _{n \in \mathbb{N}}$ from $p$ to $\gamma '$ such that the index of $\lambda _n$ goes to infinity for $n \to \infty$. We will now show that the $\tau _n$, as defined in the proposition, cannot be bounded below, i.e., that there is a subsequence such that $\tau _n \to - \infty$. This will also imply that $\tau _n < \tau _b$ for almost all $n$, i.e., that almost all $\lambda _n$ arrive at $\gamma$. By contradiction, let us assume that there is a lower bound for the $\tau _n$. As there is obviously an upper bound for the $\tau _n$, given by the time coordinate of the event $p$, this would imply that the $\tau _n$ are confined to a compact interval, so there would be a subsequence of the sequence $ \big( \gamma ' ( \tau _n ) \big)$ that converges to a point $q$ on $\gamma '$. This would give us a converging sequence of points that lie on the timelike curve $\gamma '$ and also on the past light-cone of $p$ which is an immersed lightlike submanifold near $q$ by assumption. This is possible only if this is a constant sequence. This would give us infinitely many mutually different lightlike geodesics $\lambda _n$ from $p$ to $q$. As there is a unique lightlike direction tangent to the light-cone at $p$, and as there are no periodic lightlike geodesics in a globally hyperbolic spacetime, this is impossible. 
\hfill $\blacksquare$
\end{proof}

Proposition \ref{pro4} makes sure that in the wormhole spacetime an observer at $p$ sees infinitely many images of a light source with worldline $\gamma$, under very mild restrictions on $\gamma$. Moreover, it implies that the past light-cone of every point $p$ must have a nonempty and rather complicated caustic because otherwise it would not be possible to find a sequence of past-pointing lightlike geodesics $\lambda_n$ from $p$ that intersect this caustic arbitrarily often for $n$ sufficiently large. 

The fact that $\tau _n \to - \infty$ means that the travel time of light goes to infinity; therefore, the images become fainter and fainter with increasing $n$. For such infinite sequences of light rays in the Schwarzschild spacetime, Ohanian \cite{Ohanian1987} has shown that the intensity decreases exponentially. The situation in wormhole spacetimes is quite similar. Typically, the two images with the shortest travel
time are brighter than all the infinitely many other ones combined. So there is no significant accumulation of photon energy at the observer. The situation is a bit different if $\gamma$ passes through the caustic of the past light-cone of $p$. (This situation had to be excluded for applications of Morse theory.) Then the gravitational field focusses the light towards the observer and the energy \emph{density} of the light at the observer may be quite high. In the ray optical approximation it is even infinite, whereas a wave-optical treatment shows that it is always finite. This was quantitatively worked out, again for the Schwarzschild spacetime, in another paper by Ohanian \cite{Ohanian1974}.

In the next proposition, we show that all past-pointing lightlike geodesics from $p$ to $\gamma$ come actually arbitrarily close to $K$.
\begin{proposition}
Let $W$ be any open subset in $M$ that contains the region $K$. Then all but finitely many past-pointing lightlike geodesics from $p$ to $\gamma$ intersect $W$.
\label{prop:Kapproach}
\end{proposition}
\begin{proof} 
We choose the same orthogonal splitting as in the proof of Proposition \ref{pro4} and denote the projection onto the time axis by $t$. Then the sequence $\big( \lambda _n \big)_{n \in \mathbb{N}}$ of lightlike geodesics from Proposition \ref{pro4} gives us a sequence of lightlike vectors $( w_n )_{n \in \mathbb{N}}$ with $dt (w_n)=-1$ at $p$ and a sequence of parameter values $(s_n)_{n \in \mathbb{N}}$ such that $\mathrm{exp} (s_n w_n)$ is on $\gamma$ for all $n \in \mathbb{N}$. The set of all lightlike vectors $w$ at $p$ with $dt (w)=-1$ form a $2$-sphere; by compactness, a subsequence of $(w_n)_{n \in \mathbb{N}}$ converges towards a lightlike vector $w_{\infty}$. Along the lightlike geodesic $s \mapsto \mathrm{exp} (s w_{\infty})$ the time coordinate $t$ must go to $- \infty$ by Proposition \ref{pro4} and the modulus of the radius coordinate $|\ell|$ must be bounded by Proposition  \ref{pro3}. This implies that along this geodesic the function $\ell$ either has a minimum and a maximum or converges towards a limit value $\ell _{\infty}$. In the first case, by Proposition  \ref{prop:minmax} the minimum must lie in $M_{\mathrm{out}}$ or in $K$ and the maximum must lie in $M_{\mathrm{in}}$ or in $K$. This implies that the geodesic intersects $K$ and, thus, $W$ because if neither the minimum nor the maximum is in $K$ then, by continuity, the geodesic must intersect $K$ between these two points. In the second case both the first and the second derivative of $\ell$ with respect to the parameter $s$ must go to 0. As we know from the proof of Proposition \ref{prop:minmax} that the implication (\ref{eq:posMin}) holds on $M_{\mathrm{in}}$ and the implication (\ref{eq:negMout}) holds on $M_{\mathrm{out}}$, the geodesic must come arbitrarily close to $K$, i.e., it must intersect $W$.
\hfill $\blacksquare$
\end{proof}

In the domain of outer communication of a Kerr-Newman black hole a statement analogous to Proposition \ref{prop:Kapproach} is true but it can actually be strengthened, see Hasse and Perlick \cite{Hasse:2005vu}: As the Hamilton-Jacobi equation is separable in the Kerr-Newman case, there is a spherical lightlike geodesic through each point of $K$ and the limiting geodesic $\lambda _{\infty}$ that is constructed in the proof of Proposition \ref{prop:Kapproach} must actually asymptotically spiral towards one of these spherical lightlike geodesics. In the wormhole spacetime it is not in general true that there are spherical lightlike geodesics and the limiting geodesic may oscillate between regions that are far apart from each other forever, see Example 2 below.    

\section{EXAMPLE 1}\label{sec:ex1}
As the first example, we consider the metric \eqref{metric} with 
\begin{equation}
\begin{aligned}
N&=1+\dfrac{(4 a \cos\vartheta)^2}{r_0^3\sqrt{\ell^2+r_0^2}},\\
R&=\sqrt{\ell^2+r_0^2}+\dfrac{(4 a \cos\vartheta)^2}{r_0^3},\\
\omega&=\dfrac{2 a }{(\ell^2+r_0^2)^{3/2}},\\
\end{aligned}
\end{equation}
where $r_0$ is a positive parameter with the dimension of a length and $a$ is a parameter with the dimension of a length squared. As we use units with $c=1$, both our time coordinate $t$ and our radial coordinate $\ell$ have the dimension of a length. 

This wormhole is symmetric with respect to a throat at $\ell =0$, with $r_0$ determining the radius of the throat in the equatorial plane. The parameter $a$ determines the angular momentum of the wormhole; for $a=0$ one gets the spherically symmetric and static Ellis wormhole \cite{Ellis1973}. 

Then 
\begin{equation}
\begin{aligned}
\Psi_{\pm} = -\dfrac{\ell^2+r_0^2 \pm 2 a \sin\vartheta}{(\ell^2+r_0^2)^{3/2}\;\sin\vartheta} \, ,
\end{aligned}
\end{equation}
hence
\begin{equation}
d\Psi_{\pm}=\dfrac{\ell(\ell^2+r_0^2 \pm 6 a\; \sin\vartheta)}{(\ell^2+r_0^2)^{5/2}\;\sin\vartheta}\;d\ell+
\dfrac{\cos\vartheta}{\sin^2\vartheta \sqrt{\ell^2+r_0^2}}\;d\vartheta \, .
\label{dpsi1}
\end{equation}

By \eqref{dpsi1}, photon circles are located where $d\Psi_{\pm}$ vanishes. We assume $a>0$ and we distinguish two cases. 

1) $6 a  \le r_0^2$:
Then $d \Psi _{\pm}$ vanishes if and only if 
\begin{equation}
\vartheta=\dfrac{\pi}{2} \quad \text{and} \quad \ell=0 \, , 
\end{equation}
i.e., we have in this case one co-rotating circular lightlike geodesic at $\ell_{+}^{\mathrm{ph}}=0$ and one counter-rotating circular lightlike geodesic at $\ell_{-}^{\mathrm{ph}}=0$, both of which are in the equatorial plane.

2) $6 a  > r_0^2$:
In this case $d \Psi _{\pm}$ vanishes if and only if  
\begin{equation}
\vartheta=\dfrac{\pi}{2} \quad \text{and} \quad \ell(\ell^2+r_0^2 \pm 6 a\; \sin\vartheta)=0 \, ,
\end{equation}
i.e., we have one co-rotating circular lightlike geodesic at $\ell_{+}^{ph}=0$ and three counter-rotating circular lightlike geodesics at $\ell_{-}^{ph} \in \Big\{0,\pm \sqrt{6 a-r_0^2}\Big\}$ all of which are in the equatorial plane.

For plotting the equipotential surfaces $\Psi _{\pm} = \mathrm{const.}$ and the regions $M_{\mathrm{in}}$, $M_{\mathrm{out}}$, $K_+$ and $K_-$ we take $e^{\ell/r_0}$ as the radial coordinate such that $\ell=-\infty$ corresponds to the origin. Admittedly, this has the slight disadvantage that the symmetry with respect to the throat, $\ell \mapsto -\ell$, is not shown in the diagrams; however, much more importantly it has the great advantage that the entire range $\ell \in \, ]-\infty,+\infty[$ is covered in one plot.
\begin{figure}[H]
  \centering
  \subfloat{\includegraphics[width=0.28\linewidth]{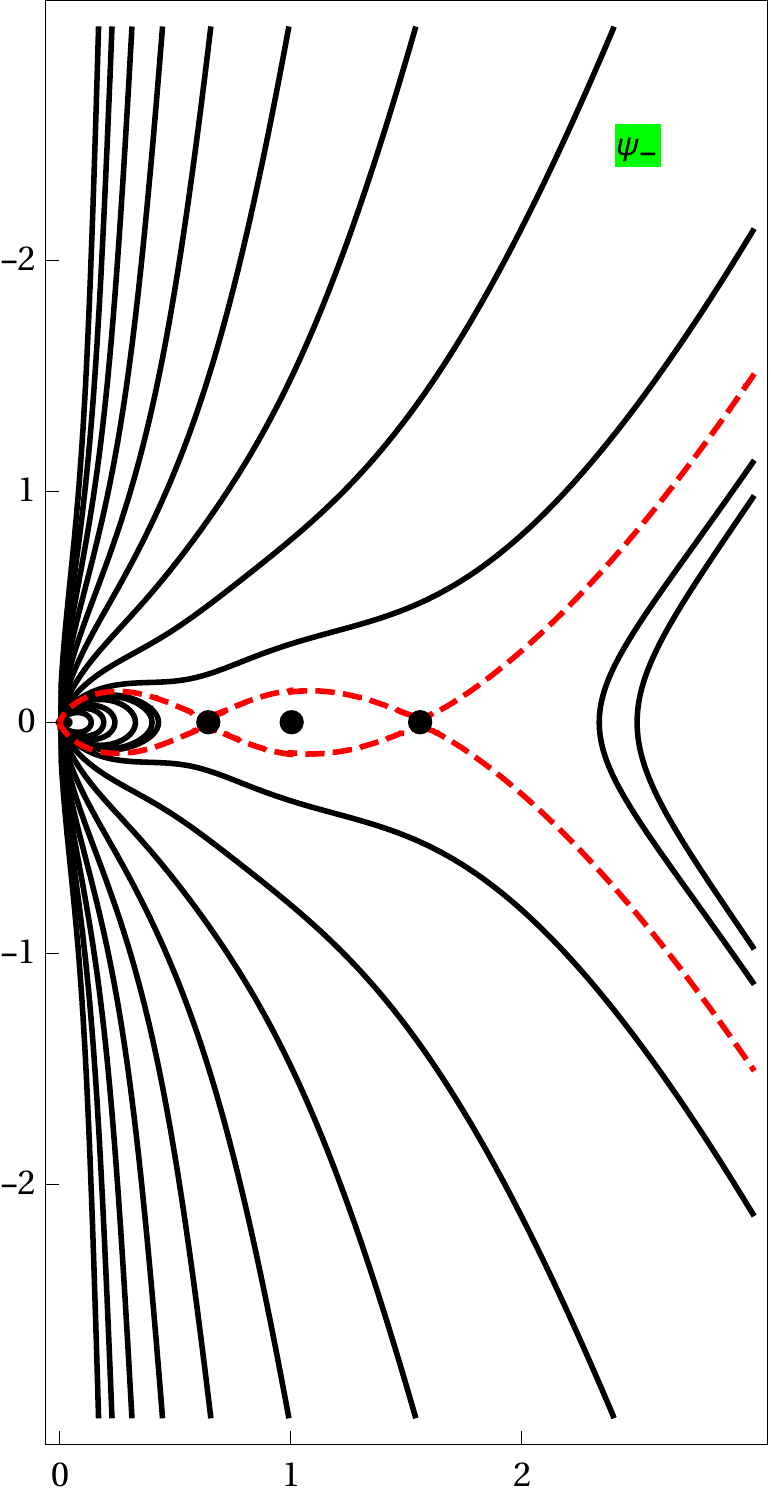}}%
  \qquad
  \subfloat{\includegraphics[width=0.28\linewidth]{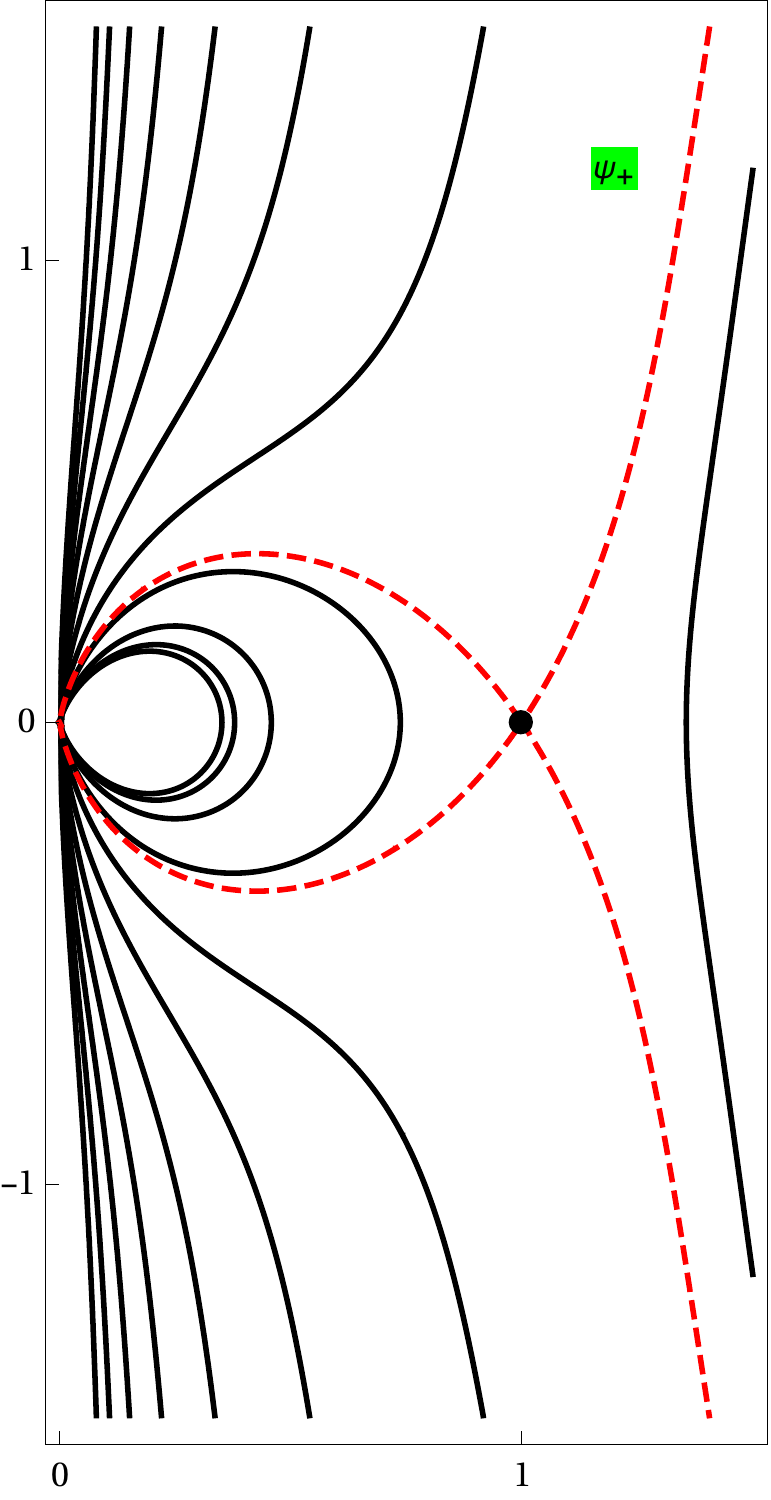}}%
   \caption{The equipotential surfaces $\Psi_{-}=\,$constant (left) and $\Psi_{+}=\,$constant (right) are drawn here for the case that $a=0.2 \, r_0^2$, i.e., $6 \, a^2 >r_0^2$. The picture shows the plane ($\varphi,t)\,$= constant, $e^{\ell/r_0} \sin\vartheta$ on the horizontal and $e^{\ell/r_0} \cos \vartheta$ on the vertical axis. The photon circles are indicated by black dots. The special equipotential surface which goes through a photon circle that is unstable with respect to radial perturbations is drawn as a dashed (red) curve.}%
   \label{psi}
\end{figure}

In Fig. \ref{psi} we show the equipotential surfaces $\Psi _{\pm} = \mathrm{constant}$ for the case that $6 \, a^2 >r_0^2$. As already mentioned above, in this case we find one co-rotating photon circle and three counter-rotating photon circles all of which are in the equatorial plane. Whereas all four photon circles are stable with respect to latitudinal perturbations, only one of them, namely a counter-rotating photon circle, is stable with respect to radial perturbations. This particular photon circle corresponds to a maximum of the potential $\Psi _-$ in the $(\ell, \vartheta)$ plane, whereas the other three photon circles are saddle-points. For each of the saddle-points, we have drawn in Fig.$\,$\ref{psi} the equipotential surface that passes through this photon surface as a dashed (red) curve. The equipotential surface passing through the stable photon circle degenerates, of course, in this picture, to a single point; neighbouring equipotential surfaces are (topological) circles in this picture, i.e., tori in 3-dimensional space.    

Fig.$\,$\ref{fig:Kex1} shows the region $K$, which is the closure of the region $K_+ \cup K_-$, for this first example. The crucial feature of the region $K$ is in the fact that each infinite sequence of past-oriented lightlike geodesics from \emph{any} point $p$ to \emph{any} generic worldline $\gamma$ in $M$ converges to a lightlike geodesic that comes arbitrarily close to $K$. As we read from the picture, in this first example the region $K$ is not very much different from the photon region in the Kerr-Newman spacetime, see Hasse and Perlick \cite{Hasse:2005vu}. The only difference is in the fact that in the wormhole case the region $K$ is separated from  the axis. We will see now in a second example that, quite generally, the region $K$ may be much more different from the Kerr-Newman case.
\begin{figure}[H]
\centering
    \includegraphics[width=0.33\textwidth]{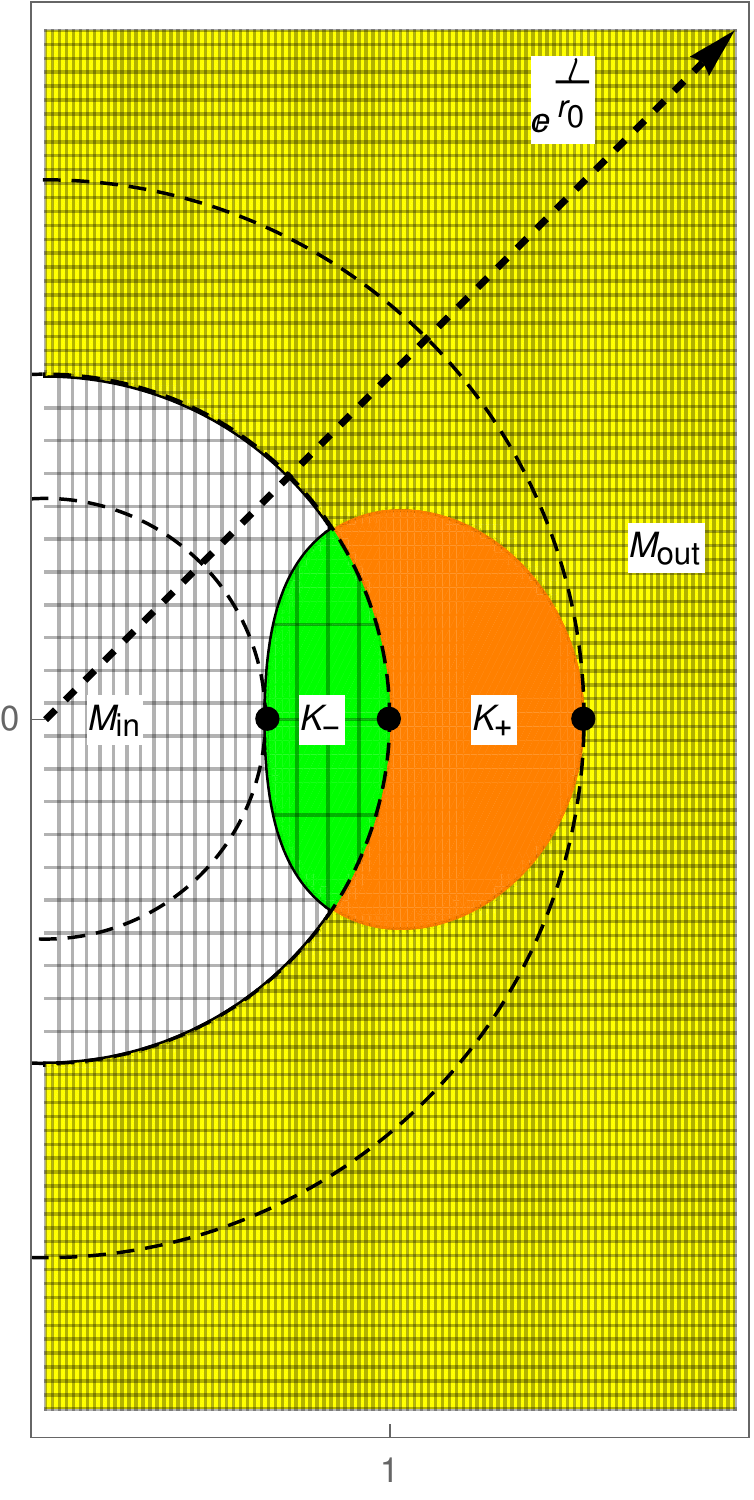}
    \caption{The regions $M_{in}$, $K_{+}$, $K_{-}$ and $M_{out}$ defined in Definition \ref{def:inout} are shown here for the case $a=0.2 r_0^2$, hence $6 a > r_0^2$. Again, we plot $e^{\ell/r_0} \sin\vartheta$ on the horizontal and $e^{\ell/r_0} \cos \vartheta$ on the vertical axis. The boundaries of $K_{+}$ and $K_{-}$  meet the equatorial plane in the photon circles which are indicated by black points.}
\label{fig:Kex1} 
\end{figure}

\section{EXAMPLE 2}\label{sec:ex2}
We give now an example where the regions $K_{+}$ and $K_{-}$ are not connected. Let 
\begin{equation}
\begin{aligned}
N&=1\\
R&=\sqrt{\ell^2+r_0^2},\\
\omega&=\dfrac{1}{r_0} \sin\Big(\dfrac{r_0^4 s}{(\ell^2+r_0^2)^2} \Big)
\end{aligned}
\end{equation}
where $r_0$ is a positive constant with the dimension of a length and $s$ is a dimensionless number. 
Then 
\begin{equation}
\Psi_{\pm} =\mp \dfrac{1}{r_0} \sin \Big(\dfrac{r_0^4\;s}{(\ell^2+r_0^2)^2}\Big)-\dfrac{1}{\sin\vartheta \sqrt{\ell^2+r_0^2}}
\, , 
\end{equation}
hence
\[
d\Psi_{\pm}= \Bigg( 
\pm\dfrac{4 \, r_0^3 \;s\;\ell }{(\ell^2+r_0^2)^3} \, \cos \Big(\dfrac{r_0^4\;s}{(\ell^2+r_0^2)^2} \Big) 
+\dfrac{\ell}{\sin\vartheta(\ell^2+r_0^2)^{3/2}} \Bigg) \;d\ell
\]
\begin{equation}
+\dfrac{\cot\vartheta}{\sin\vartheta \sqrt{\ell^2+r_0^2}}\;d\vartheta \, . 
\end{equation}
From this expression we read that photon circles are located at 
\begin{equation}
\vartheta = \dfrac{\pi}{2} \quad \text{and} \quad
\ell \Bigg( \pm4 \, r_0^3 \;s \, \cos \Big(\dfrac{r_0^4\;s}{(\ell^2+r_0^2)^2} \Big) 
+
(\ell^2+r_0^2)^{3/2} \Bigg) = 0 \, . 
\end{equation}
So there is always one co-rotating and one counter-rotating photon circle at $\ell =0$ in the equatorial plane and, depending on $s$, a certain number of additional photon circles, both co-rotating and counter-rotating, in the equatorial plane whose radius coordinates $\ell$ are given by a transcendental equation.  

In the following plots we use $e^{\ell/r_0}$ as the radial coordinate, as in Example 1. In Figs. \ref{psi2} and \ref{psi3} we show the potentials $\Psi _{\pm}$ for the second example with  $s=2$. In this case we have three co-rotating and five counter-rotating photon circles. Again, all photon circles are stable with respect to latitudinal perturbations. Two of the co-rotating and three of the counter-rotating ones are unstable with respect to radial perturbations (saddle-points of the respective potential), the other ones are stable with respect to radial perturbations (maxima of the respective potential). As before, photon circles that are unstable with respect to radial perturbations correspond to self-intersections of equipotential surfaces which are indicated by dashed (red) curves. By choosing larger values of $s$  we may have as many photon circles as we like.

In Fig. \ref{fig:Kex2} we show the regions $K_+$ and $K_-$ for the second example with $s=2$. In this case $K_+$ and $K_-$ have two connected components each. By choosing a bigger value for $s$ the regions $K_+$ and $K_-$ may have as many connected components as we like. We have proven in the preceding section that any infinite sequence of past-oriented lightlike geodesics from an event $p$ to a timelike curve $\gamma$ converges to a limiting lightlike geodesic that comes arbitrarily close to the region $K$ which is the closure of the union of $K_+$ and $K_-$. What this example demonstrates is the fact that this region need not be connected, and may actually have arbitrarily many connected components that may be far apart from each other. Therefore, our general result does not exclude the case that the limiting lightlike geodesic oscillates forever between two regions that are far apart from each other. This is a major difference in comparison to the spacetime of a Kerr-Newman  black hole where the limiting lightlike geodesic necessarily spirals towards a spherical lightlike geodesic.  

\begin{figure}[H]
  \centering
  \subfloat{\includegraphics[width=0.29\linewidth]{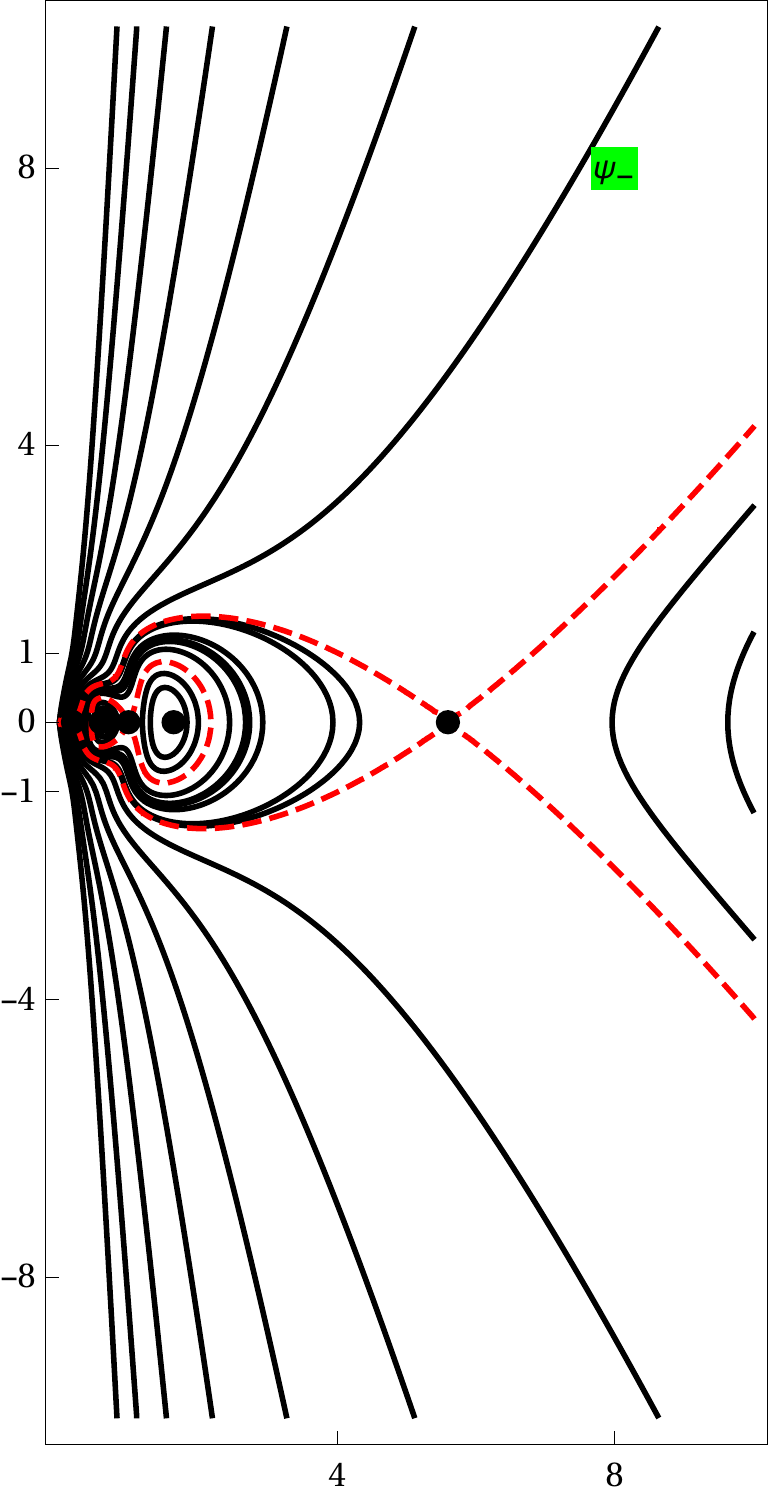}}%
  \qquad
  \subfloat{\includegraphics[width=0.29\linewidth]{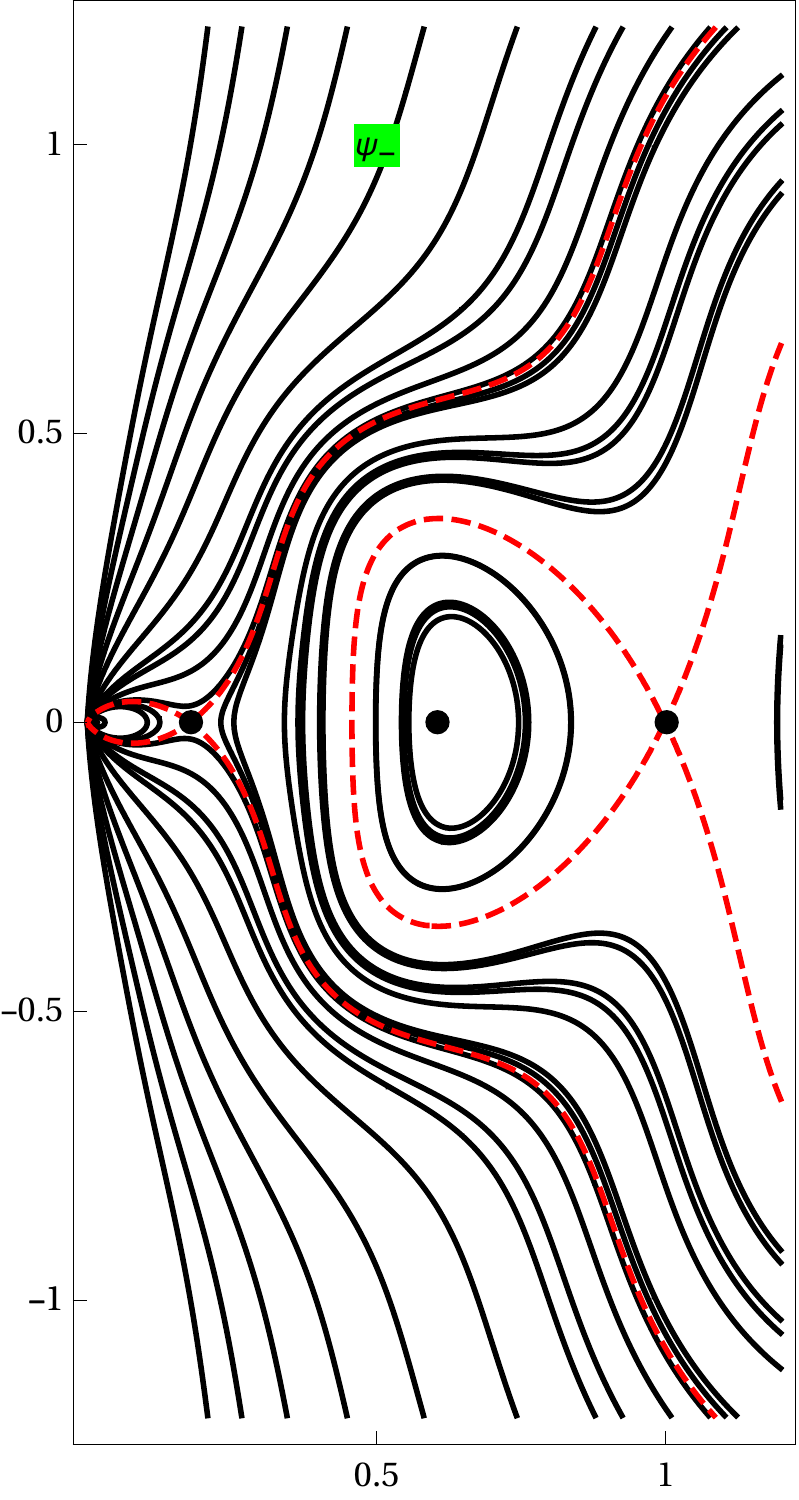}}%
   \caption{The surfaces $\Psi_{-}= \,$constant are drawn here for the case $s=2$. The picture on the right shows an enlarged version of the interior part. The photon circles are indicated by black dots, and the equipotential surface that goes through a photon circle that is unstable with respect to radial perturbations is drawn as a dashed (red) curve.}%
  \label{psi2}
\end{figure}

\begin{figure}[H]
\centering
   \includegraphics[width=0.27\linewidth]{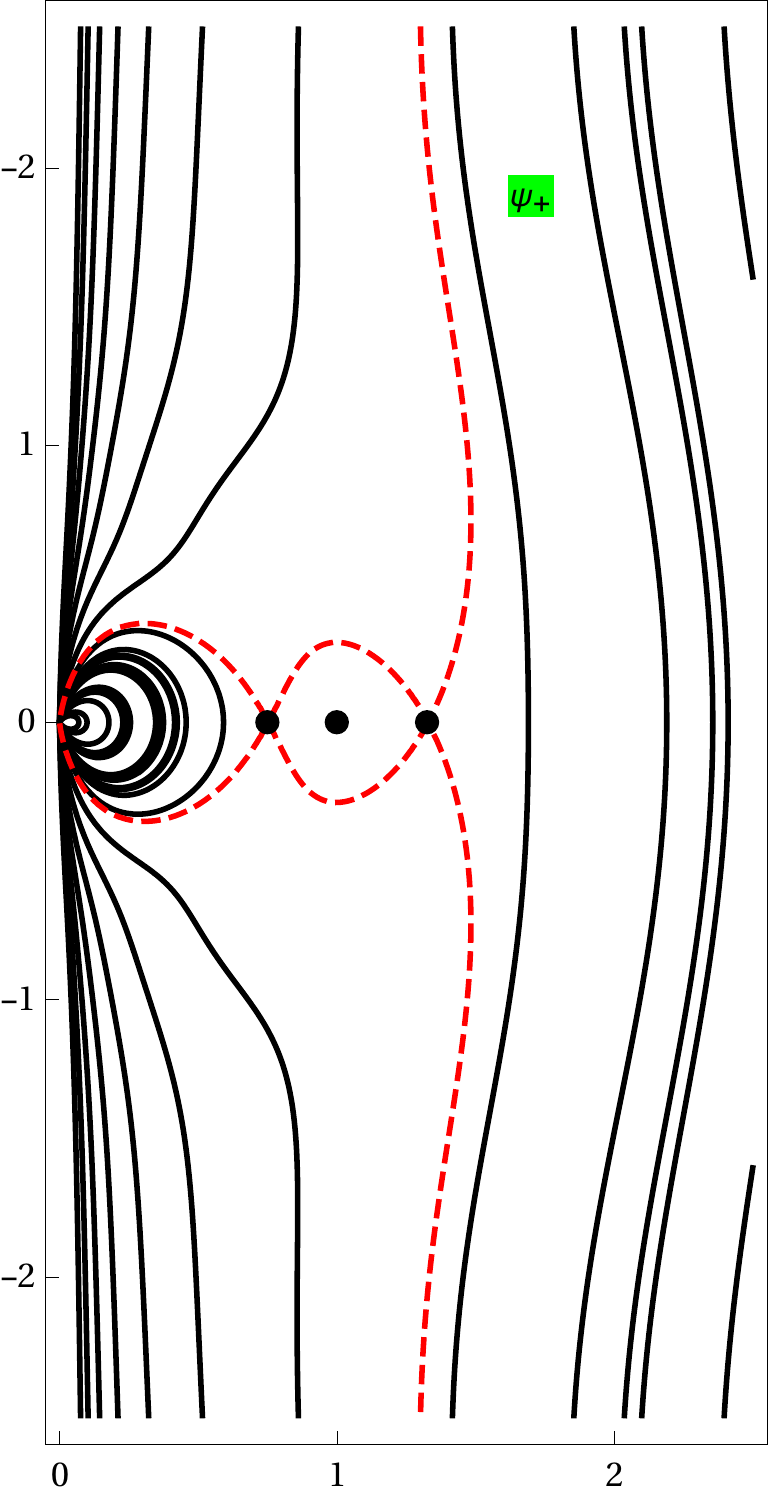}
   \caption{The surfaces $\Psi_{+}= \,$constant are drawn here for the case $s=2$, in analogy to Figure \ref{psi2}.}%
   \label{psi3}
\end{figure}

\begin{figure}[H]
\centering
    \includegraphics[width=0.55\textwidth]{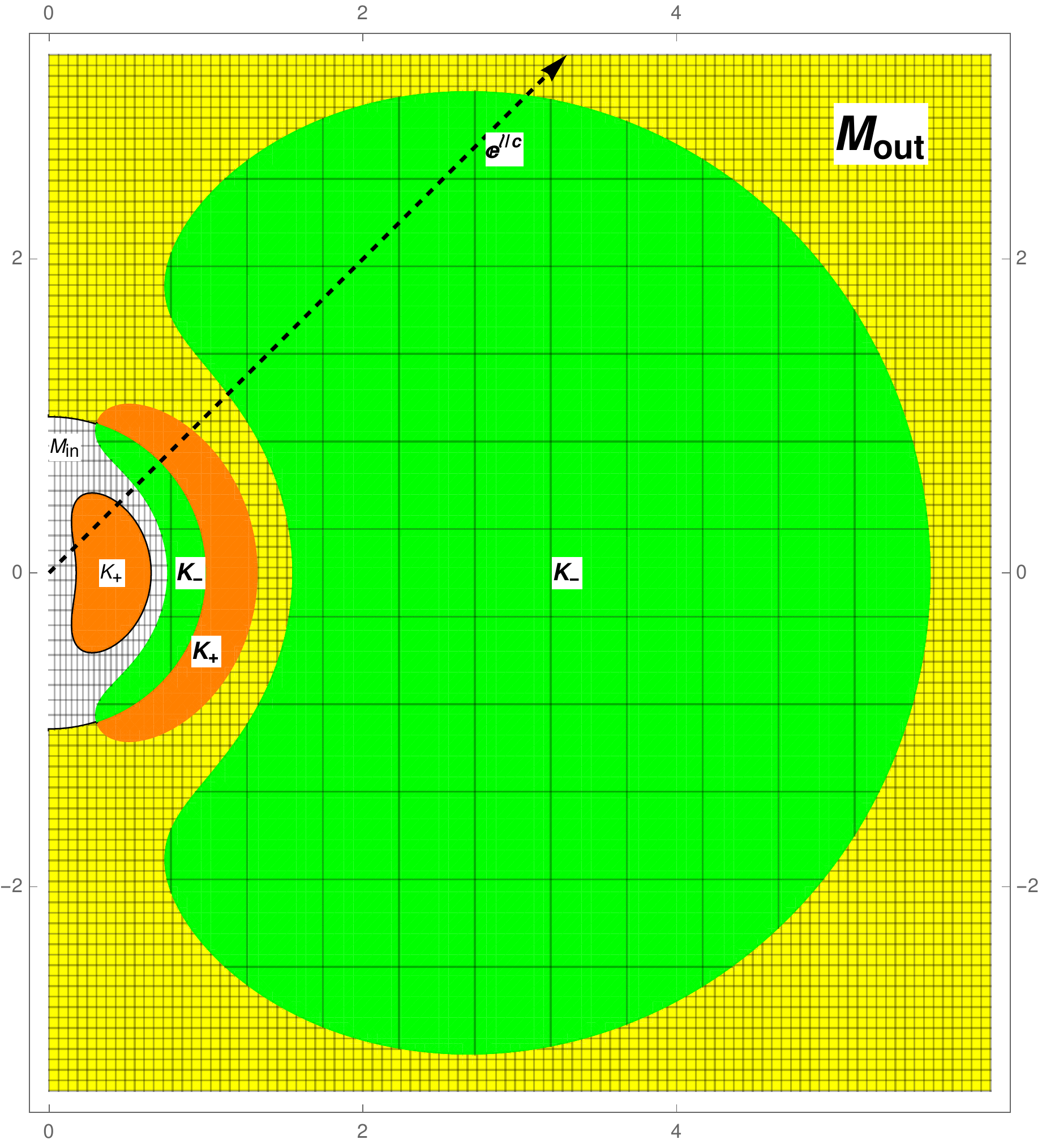}
    \caption{The regions $M_{in}$, $K_{+}$, $K_{-}$ and $M_{out}$ defined in Definition \ref{def:inout} are shown here for the case $s=2$. Again, we plot $e^{\ell/r_0} \sin\vartheta$ on the horizontal and $e^{\ell/r_0} \cos \vartheta$ on the vertical axis. The photon circles are located where the boundaries of $K_+$ and $K_-$ meet the equatorial plane. By choosing $s$ bigger, $K_+$ and $K_-$ may have arbitrarily many connected components; correspondingly, one may have arbitrarily many photon circles.
}
\label{fig:Kex2}
\end{figure}

\section{CONCLUDING REMARKS}\label{sec:conclusion}
In this paper we have considered a class of rotating traversable wormholes and we have proven, with the help of Morse theory, that in these wormhole spacetimes an observer sees infinitely many images of a light source, under very mild restrictions on the motion of the light source. In this respect wormholes are similar to Kerr-Newman black holes (and other black holes). As our Morse-theoretical approach demonstrates, this similarity has its origin in the fact that both the wormhole spacetime and the domain of outer communication of a Kerr-Newman black hole is a globally hyperbolic spacetime with topology $S^2 \times \mathbb{R}^2$ that satisfies the metric growth condition. Moreover, both in the wormhole spacetime and in the domain of outer communication of a Kerr-Newman black hole there are potentials $\Psi_+$ and $\Psi_-$ which tell us where the radial coordinate may have turning points along a lightlike geodesic. However, there are also important differences. In the case of a Kerr-Newman black hole there is a photon region filled with lightlike geodesics each of which stays on a sphere $r = \mathrm{constant}$. If we consider an infinite sequence of lightlike geodesics from an event $p$ to a generic timelike curve $\gamma$, they converge towards a lightlike geodesic $\lambda _{\infty}$ that asymptotically spirals towards one of these ``spherical'' lightlike geodesics that fill the photon region. As shown by our examples, the situation can be much more complicated in the wormhole spacetimes. In general, there are no spherical lightlike geodesics in the wormhole spacetime. The natural generalisation of the photon region, denoted $K$ in this paper, is the closure of two open sets, $K_-$ and $K_+$, each of which may consist of arbitrarily many connected components. The above-mentioned lightlike geodesics $\lambda _{\infty}$ have to come close to the region $K$, as we have proven, but they need not spiral towards a certain limit curve; e.g., they may oscillate between different connected components of $K$ forever. 

As Uhlenbeck's theorem does not require stationarity or
axisymmetry, we expect that the existence of infinitely many images will hold true also for wormholes without any symmetry, as long as global hyperbolicity and the metric growth condition are still satisfied. Other future applications of Uhlenbeck's theorem could be to globally hyperbolic spacetimes with topologies other than $S^2 \times \mathbb{R}^2$. Also, we mention that Giannoni et al \cite{GiannoniEtAl1998} have proven a theorem similar to Uhlenbeck's for lightlike geodesics in spacetimes that need not be globally hyperbolic. Although rather sophisticated, using infinite-dimensional Hilbert manifolds, we believe that the work of Giannoni et al. has the potential of giving very strong and interesting new results on lensing. In a slightly different vein, it should also be possible to establish theorems similar to the ones by Uhlenbeck or Giannoni et al. for light rays other than lightlike geodesics in a general-relativistic spacetime. E.g., one could consider the case of light rays in a plasma on a general-relativistic spacetime, or of lightlike geodesics in a Finsler spacetime. Quite generally, we believe that the potential applications of Morse theory to gravitational lensing are still in the fledgling stages. 
\begin{flushleft}
\textbf{ACKNOWLEDGMENTS}\\
\end{flushleft}
We gratefully acknowledge support from the DFG within the Research Training Group 1620 “Models of Gravity”.
\bibliographystyle{spphys}

\end{document}